\documentclass[peerreview,onecolume]{IEEEtran}
\usepackage[dvips]{graphicx}

\usepackage{amssymb}
\usepackage{amsmath}
\usepackage{color}
\usepackage{ booktabs}
\usepackage{array}  
\usepackage{research5IEEE}

\usepackage{amsfonts}
\usepackage{amssymb}
\usepackage{array}

\usepackage{amsthm}

\newtheorem{Theorem}{Theorem}
\newtheorem{Remark}{Remark}

\ifCLASSINFOpdf

\else

\fi

\begin{document}
\title{ Capacity Region of  Degraded Relay Broadcast Channel}
\author{
  \IEEEauthorblockN{Ke Wang, Youlong Wu, Yingying Ma\\}

}
\maketitle

 \begin{abstract}
 The  relay broadcast channel (RBC) is considered, in which a transmitter communicates with two receivers with the assistance of a relay. Based on different degradation orders among the relay and the receivers' outputs, three types of  physically degraded RBCs (PDRBCs)  are  introduced. Inner bounds and outer bounds are derived on the capacity region of  the presented three types. The bounds are tight for two types of PDRBCs: 1) one receiver's output is a degraded form of the other receiver's output, and the relay's output is a degraded form of the weaker receiver's output; 2)  one receiver's output is a degraded form of the relay's output, and the other receiver's  output is a degraded form of the relay's output.  For the Gaussian  PDRBC, the bounds match, i.e., establish its capacity region.
 

 \end{abstract}

\section{Introduction}
The relay channel \cite{Meulen'74}  describes a 3-node communication channel where the transmitter communicates a message to the receiver with the assistance of a relay.  The capacity of relay channel has been studied in \cite{Cover'79} by Cover and El Gamal, who developed two fundamental relay strategies: compress-forward and decode-forward. It shows that when the relay's output is a degraded form of receiver's output,  letting relay send a constant symbol achieves the channel capacity; when  the receiver's output is a degraded form of relay's output,  decode-forward strategy achieves the channel capacity. The capacity  of general relay channel, unfortunately, has not yet been found. 

In the relay channel, if the relay node also wants to decode a private message sent by the transmitter, then this channel turns to be \emph{partially cooperative} relay broadcast channels (RBCs) \cite{Liang'07Veeravalli}, \cite{Liang'07Kramer}.  The capacity region of partially cooperative RBC is established for the case of degraded message sets \cite{Liang'07Kramer},  where the transmitter has a common message for both destinations and a private message for the relay. The \emph{fully cooperative} RBC \cite{Liang'07Veeravalli} is a more general model where both destinations can serve as relay and receiver.  The fully and cooperatively cooperative RBC with  feedback was  studied in \cite{Liang'07Veeravalli,Wu'ISIT16}. 



 Another RBC model, called the dedicated RBC model was considered in  \cite{Dabora'06}, \cite{Kramer'05},  where a relay node  assists the cooperation between  two-receiver broadcast channel.   The capacity region of the dedicated RBC is generally unknown, even for the  physically degraded  case. In \cite{Bha'08},  the dedicated Gaussian  RBC is studied and capacity region is established when one receiver's output is a degraded form of the other receiver's output, and the stronger receiver's output is a degraded form of the relay's output. 
 
 In this paper, we consider the 4-node dedicated physically degraded RBC (PDRBC). Based on different degradation orders among the relay and the receivers' outputs, three types of     degraded RBCs  are  introduced. Inner bounds and outer bounds are derived on the capacity region of  presented three types. The bounds are tight when  one receiver's output is a degraded form of the other receiver's output, and the relay's output is a degraded form of either the stronger or the weaker receiver's output. For  Gaussian PDRBC, our bounds always match, i.e., establish the capacity region of Gaussian PDRBC.


{The rest of the paper is organized as follows. System model is introduced in Section \ref{sec:model}. In Section \ref{Sec:Mresults}, we state our main capacity results for PDRBC.  The proof of inner  and the outer bounds on the capacity region of DM-PDRBC  are stated in Section \ref{sec:InType1} and \ref{sec:OutTDM}, respectively. The proof of capacity region for Gaussian PDRBC is presented in Section \ref{Sec:proofGaussian}}.

Notation: We use capital letters to denote  random variables and small letters for their realizations, e.g. $X$ and $x$. Define a function ${C}(x):=\frac{1}{2}\log_2(1+x)$.
 For nonnegative integers $k,j$,  let $X^j_k:= (X_{k,1},\ldots,X_{k,j})$.  Given a distribution $p_A$ over some alphabet $\set{A}$, a positive real number $\epsilon>0$, and a positive integer $n$, $\set{T}_{\epsilon}^{(n)}(p_A)$ is the typical set in \cite{Gamal'book}.   





\section{System model}\label{sec:model}
Consider a $4$-node discrete memoryless relay broadcast channel (DM-RBC) in which there is one  
transmitter, two receivers and one relay that helps the transmitter communicate  with the receivers, as depicted in Fig. \ref{fig:Mrelays}. This channel consists of five  finite alphabets ($\set{X},\set{X}_3, \set{Y}_1,\set{Y}_2, \set{Y}_3$) and  a collection of probability mass function  (pmf) $p(y_1,y_2,y_3|x,x_3)$ on $\set{Y}_1\times\set{Y}_2\times\set{Y}_3$, one for each $\set{X}\times\set{X}_3$. Here  $x\in\set{X}$ is the input to the transmitter,  $x_3\in\set{X}_3$ is the input to the relay, $y_3\in\set{Y}_3$ is the relay's  output , and  $y_k\in\set{Y}_k$  is receiver $k$'s output, for $k\in\{1,2\}$.

\begin{figure}[ht]
\centering
\includegraphics[width=0.43\textwidth]{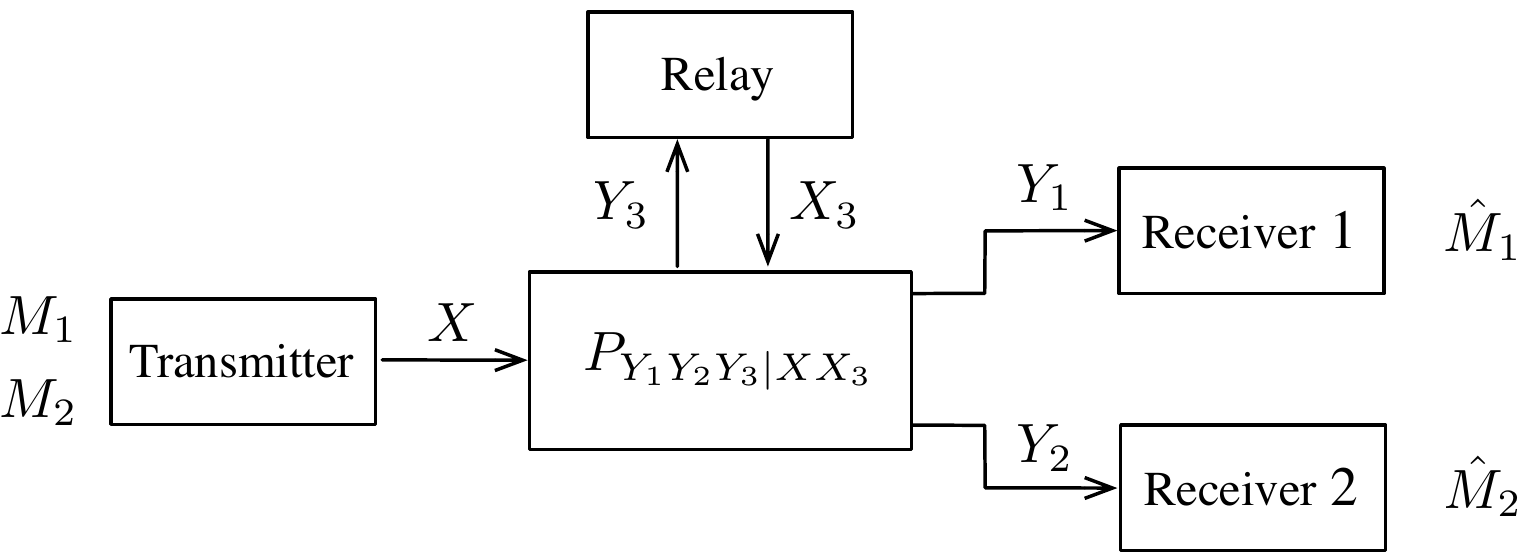}
\caption{Relay broadcast channel} \label{fig:Mrelays}
\end{figure}

The transmitter wishes to communicate a message $M_k\in[1:2^{nR_k}]$ to receiver $k$, for $k\in\{1,2\}$, with the assistance of a relay, where $n$  denotes  block length of  transmission. A $(2^{nR_1},2^{nR_2},n)$  code for this channel has
\begin{itemize}
\item two message sets $\set{M}_1=[1:2^{nR_1}]$ and $\set{M}_2=[1:2^{nR_2}]$,
\item a source encoder that maps messages $(M_1,M_2)$ to  channel input $X_{i}(M_1,M_2)$, for each time $i\in[1:n]$,
\item a relay encoder that  maps $Y^{i-1}_3$ to a sequence $X_{3,i}(Y^{i-1}_3)$, for   $i\in[1:n]$,
\item  two decoders that  estimate $\hat{M}_{1}$ and $\hat{M}_{2}$ based on $Y^n_1$ and $Y^n_2$, respectively.
\end{itemize}

Suppose $M_k$ is uniformly distributed over the message set $\set{M}_k$.  A rate region $(R_1,R_2)$  is called achievable if for every blocklength $n$, there exists  a $(2^{nR_1},2^{nR_2},n)$ code such that the average probability of error \[P^{(n)}_e=\text{Pr}{[(\hat{M}_1,M_2)\neq (M_1,M_2)]}\] tends to 0 as  $n$ tends to infinity. The capacity region $\set{C}$  is the closure of the set of all achievable rate pairs $(R_1,R_2)$. 

In this paper, we mainly focus on the PDRBC. Let $\set{C}_\text{PD}$ denote the capacity region of PDRBC. Without loss of generality, assume $Y_2$ is a random degradation of $Y_1$. According to the degradation order among $Y_1,Y_2$ and $Y_3$, we define three types of DM-PDRBC  described as below.

\begin{itemize}
\item Type-I PDRBC: \\$X-X_3Y_3-Y_1-Y_2$ forms Markov chain.
\item Type-II PDRBC:\\ $X-X_3Y_1-Y_3-Y_2$ forms Markov chain.
\item Type-III PDRBC:\\ $X-X_3Y_1-Y_2-Y_3$ forms Markov chain.
\end{itemize}

\subsection{Gaussian PDRBC}\label{sec:type1Gaussian}
Consider the  Gaussian RBC, which can be described as
\begin{IEEEeqnarray*}{rCl}
Y_i&=&X+X_3+Z_i, \quad i=1,2,3
\end{IEEEeqnarray*}
where ${Z}_1,Z_2$, and $Z_3$ are  Gaussian noise components with zero mean and variances $\sigma_1^2, \sigma_2^2$ and $\sigma_3^2$, respectively. Assume average transmission power constraint $P$ on the transmitter, and $P_r$ on the relay.

Similar to the discrete memoryless case, Gaussian PDRBC can be divided into three types according to the degradation order among the outputs received at the receivers and relay:
\begin{itemize}
\item Type-I Gaussian PDRBC: $X-X_3Y_3-Y_1-Y_2$ forms Markov chain, which is equivalent to
\begin{IEEEeqnarray*}{rCl}
Y_3&=&X+X_3+Z_3,\nonumber\\
Y_1&=&X+X_3+Z_1=X+X_3+Z_3+\hat{Z}_{\textnormal{a}},\nonumber\\
Y_2&=&X+X_3+Z_2=X+X_3+Z_3+\hat{Z}_{\textnormal{a}}+\tilde{Z}_{\textnormal{a}}´		
\end{IEEEeqnarray*}
where $\hat{Z}_{\textnormal{a}}\sim\set{N}(0,\sigma_1^2-\sigma_3^2)$ and $\tilde{Z}_{\textnormal{a}}\sim\set{N}(0,\sigma_2^2-\sigma_1^2)$  are independent.
 \item Type-II Gaussian PDRBC:  $X-X_3Y_1-Y_3-Y_2$ forms Markov chain, which is equivalent to
 \begin{IEEEeqnarray*}{rCl}
Y_1&=&X+X_3+Z_1,\nonumber\\
Y_3&=&X+X_3+Z_3=X+X_3+Z+\hat{Z}_{\textnormal{b}},\nonumber\\
Y_2&=&X+X_3+Z_2=X+X_3+Z+\hat{Z}_{\textnormal{b}}+\tilde{Z}_{\textnormal{b}}
\end{IEEEeqnarray*}
where $\hat{Z}_{\textnormal{b}}\sim\set{N}(0,\sigma_3^2-\sigma_1^2)$, and $\tilde{Z}_{\textnormal{b}}\sim\set{N}(0,\sigma_2^2-\sigma_3^2)$  are independent. 
\item Type-III Gaussian PDRBC: $X-X_3Y_1-Y_2-Y_3$ forms Markov chain, which is equivalent to 
\begin{IEEEeqnarray*}{rCl}
Y_1&=&X+X_3+Z_1,\nonumber\\
Y_2&=&X+X_3+Z_2=X+X_3+Z_1+\hat{Z}_{\textnormal{c}},\nonumber\\
Y_3&=&X+X_3+Z_3=X+X_3+Z_1+\hat{Z}_{\textnormal{c}}+\tilde{Z}_{\textnormal{c}}
\end{IEEEeqnarray*}
where $\hat{Z}_{\textnormal{c}}\sim\set{N}(0,\sigma_2^2-\sigma_1^2)$, and $\tilde{Z}_{\textnormal{c}}\sim\set{N}(0,\sigma_3^2-\sigma_2^2)$  are independent.
\end{itemize}

\section{Capacity Results for DM-PDRBC}\label{Sec:Mresults}
This section  presents our main results. The  proofs are given in Sections \ref{sec:Ach},  \ref{sec:OutTDM} and \ref{Sec:proofGaussian}.
\begin{Theorem}\label{Them:type1}
 For  Type-I PDRBC, the inner bound on the capacity region consists of all rate pairs $(R_1,R_2)$ such that 
\begin{subequations}\label{eq:regionThm1}
  \begin{IEEEeqnarray}{rCl}
  R_2&\leq& I(U,V;Y_2),\\
R_1+R_2&\leq& I(U,V;Y_2)+I(X;Y_3|U,V,X_3),\quad\\
R_1+R_2&\leq& I(U,V;Y_2)+I(X,X_3;Y_1|U,V),\\
R_1+R_2&\leq& I(X;Y_3|V,X_3)
\end{IEEEeqnarray}
\end{subequations}
for some pmf $p(u,v)p(x_3|v)p(x|u,x_3)$. The outer bound  on the capacity region has  same rate constraints as \eqref{eq:regionThm1}, but under the pmf $p(u,v)p(x_3|u)p(x|u,x_3)$.
\end{Theorem}

\begin{proof}
See the achievability in Section \ref{sec:InType1}. The proof of the outer bound is given in Section \ref{sec:OutType1}.
\end{proof}


\begin{Remark} In \cite[Lemma 1]{Bha'08}, it shows that when $p(y_1,y_2,y_3|x,x_3)=p(y_3|x,x_3) p(y_1,y_2|y_3,x_3)$, the capacity region of RBC depends only depends only on the marginal distributions $p(y_1|x_3,y_3)$ and $p(y_2|x_3,y_3)$. Thus the capacity region of Type-I PDRBC holds when it is   stochastically degraded. 
\end{Remark}

\begin{Theorem}\label{Them:type2}
 For  Type-II PDRBC, the capacity region is the set of rate pairs $(R_1,R_2)$ such that 
 \begin{subequations}\label{eq:regionThm2}
  \begin{IEEEeqnarray}{rCl}
R_1&\leq& I(X;Y_1|U,X_3)\\
R_2&\leq & \min \left\{I(U,X_3;Y_2),I(U;Y_3|X_3)\right\}
\end{IEEEeqnarray}
\end{subequations}
for some pmf $p(u,x_3)p(x|u)$. 
\end{Theorem}

\begin{proof}
See the achievability  in Section \ref{sec:InType2}. The converse is given in Section \ref{sec:OutType2}.
\end{proof}

\begin{Theorem}\label{Them:type3}
 For  Type-III PDRBC, the capacity region is the set of rate pairs $(R_1,R_2)$ such that 
 \begin{subequations}
  \begin{IEEEeqnarray}{rCl}\label{eq:regionThm2}
R_1&\leq& I(X;Y_1|U,X_3=x_3)\\
R_2&\leq &  I(U;Y_2|X_3=x_3),
\end{IEEEeqnarray}
\end{subequations}
for some value $x_3\in\set{X}_3$ and pmf $p(u,x)$. 
\end{Theorem}

\begin{proof}
In the achievability, let the relay send a constant value $x_3$, and the transmitter use the traditional superposition coding  to send the source messages, i.e., the weak receiver's message stored in a cloud center codeword $u^n$  is decoded by both receivers, and the strong receiver's message conveyed through a satellite codeword $x^n$ is only decoded by the strong receivers. The converse  is given in Section \ref{sec:OutType3}.
\end{proof}


  
\begin{Theorem}\label{Them:Gaussian}
 For the  Type-I Gaussian PDRBC, the inner bound of Theorem \ref{Them:type1} is tight and the capacity region is the set of rate pairs $(R_1,R_2)$ such that
\begin{subequations}\label{region:T1}
\begin{IEEEeqnarray}{rCl}
&&R_1\leq \min  \left\{C\left(\bar{\alpha}\theta\frac{P}{\sigma_3^2}\right),\right.\nonumber\\ && \left.\hspace{1.5cm} C\left( \frac{\theta P+\theta_r P_r+2\sqrt{\alpha\theta \theta_r  PP_r}}{\sigma_1^2} \right)\right\},\\
&&R_2\leq C\left(\frac{\bar{\theta}P+\bar{\theta}_rP_r+2\sqrt{\beta\bar{\theta}\bar{\theta}_rPP_r}}{\theta P+\theta_r P_r+2\sqrt{\alpha\theta \theta_r  PP_r}+\sigma_2^2}\right),\quad\\
&&R_1+R_2 \leq C\left(\frac{\bar{\beta}\bar{\theta}P+\bar{\alpha}\theta  P}{\sigma^2_3}\right)
\end{IEEEeqnarray}
\end{subequations}
where $0\leq\alpha,\beta,\theta,\theta_r\leq1$. 
\item For the  Type-II Gaussian PDRBC, the  capacity region is the set of rate pairs $(R_1,R_2)$ such that
\begin{subequations}\label{region:T2}
\begin{IEEEeqnarray}{rCl}
R_1&\leq& C\left(\frac{\alpha P}{\sigma_1^2} \right),\\
R_2&\leq& \min\left\{ C\left(\frac{\bar{\alpha}P+P_r}{\alpha P+\sigma^2_2} \right) ,C\left(\frac{(\beta-{\alpha})P}{\alpha P+\sigma^2_3} \right)\right\} \quad 
\end{IEEEeqnarray}
\end{subequations}
where $0\leq \rho_1,\rho_2,\alpha,\beta\leq 1$ and $\beta\geq{\alpha}$. 

\item  For the  Type-III Gaussian PDRBC,  the capacity region is the set of rate pairs $(R_1,R_2)$ such that
\begin{subequations}\label{region:T3}
\begin{IEEEeqnarray}{rCl}
R_1&\leq& C\left(\frac{\alpha P}{\sigma_1^2} \right),~
R_2\leq C\left(\frac{\bar{\alpha}P}{\alpha P+\sigma^2_2} \right) 
\end{IEEEeqnarray}
\end{subequations}
where $0\leq \alpha\leq 1$.
\end{Theorem}
\begin{proof}
See Section \ref{Proof:GaussianT1}, \ref{sec:type2Gaussian} and \ref{sec:type3Gaussian}.
\end{proof}


\section{Coding Schemes for DM-PDRBC}\label{sec:Ach}

\subsection{Inner bound for Type-I PDRBC}\label{sec:InType1}
We present a block-Markov coding scheme that consists of $B+1$ blocks, where messages $M_{1,b}\in[1:2^{nR_1}]$ and $M_{2,b}\in[1:2^{nR_2}]$, for $b\in[1:B]$, are sent to the receivers over $B+1$ blocks.  Split message $M_{1,b}$ into $(M'_{1,b},M''_{1,b})$, where $M'_{1,b}\in[1:2^{nR_1'}]$ and $M''_{1,b}\in[1:2^{nR_1''}]$ are independent with each other. Let messages $M_{1,B+1}=M_{2,B+1}=1$ and assume  $(M_{1,B+1},M_{2,B+1})$  are known by the relay and receivers before communication.

\subsubsection{Codebook}
Fix the pmf $p(v)p(u|v)p(x_3|v)p(x|u,x_3)$.
For each block $b\in[1:B+1]$, randomly and independently generate $2^{n(R_{2}+{R}'_{1})}$ sequences $v_{b}^n(m_{2,b-1},m'_{1,b-1})\sim \prod^n_{i=1}p_{V}(v_{b,i})$.  For each $(m_{2,b-1},m'_{1,b-1})$, randomly and independently generate $2^{n{R}''_1}$ sequences ${x}_{3,b}^n(m''_{1,b-1}|m_{2,b-1},m'_{1,b-1})\sim \prod^n_{i=1}p_{X_3|V}(x_{3,b,i}|v_{b,i})$. For each $(m_{2,b-1},m'_{1,b-1})$, randomly and independently generate $2^{n({R}'_1+R_2)}$ sequences ${u}_{b}^n(m_{2,b},m'_{1,b}|m_{2,b-1},m'_{1,b-1})\sim\prod^n_{i=1}p_{U|V}(u_{b,i}|v_{b,i})$. For each $(m_{2,b},m'_{1,b},m'_{1,b-1},m''_{1,b-1},m_{2,b-1})$, randomly and independently generate $2^{n{R}''_1}$ sequences 
  \begin{equation*}
  \begin{aligned}
  x_{b}^n(m''_{1,b}| m_{2,b},m'_{1,b},m''_{1,b-1},m_{2,b-1},m'_{1,b-1})\\\sim \prod^n_{i=1}p_{X|X_3U}(x_{b,i}|x_{3,b,i},u_{b,i}). 
\end{aligned}
\end{equation*}

\subsubsection{Transmitter's encoding}
  In each block $b\in[1:B+1]$, the transmitter  sends \[x_{b}^n(m''_{1,b}| m_{2,b},m'_{1,b},m''_{1,b-1},m_{2,b-1},m'_{1,b-1}).\]


\subsubsection{Relay's encoding}
In each block $b\in[1:B+1]$, the relay decodes $u_b^n$ and $x_b^n$, for $b\in[1:B]$, by looking for a tuple of messages $(\hat{m}'_{1,b},\hat{m}_{2,b},\hat{m}''_{1,b})$ such that 
\begin{IEEEeqnarray}{rCl}\label{eq:Sch1CF}
&&\big(u_b^n(\hat{m}_{2,b},\hat{m}'_{1,b}|m_{2,b-1},m'_{1,b-1}),\nonumber\\&&  v_{b}^n(m_{2,b-1},m'_{1,b-1}), x_{3,b}^n(m''_{1,b-1}|m_{2,b-1},m'_{1,b-1}),\nonumber\\&& 
x_{b}^n(\hat{m}''_{1,b}| \hat{m}_{2,b},\hat{m}'_{1,b},m''_{1,b-1},m_{2,b-1},m'_{1,b-1}),y_{3,b}^n
\big) \nonumber\\&& \in\mathcal{T}^n_{\epsilon/2}(p_{UVXX_3Y_3}).\nonumber
\end{IEEEeqnarray}
Then, it sends $x^n_{3,b+1}(\hat{m}''_{1,b}|\hat{m}_{2,b},\hat{m}'_{1,b})$ in block $b+1$.

By the independence of the codebooks, the Markov lemma \cite{Gamal'book}, packing lemma \cite{Gamal'book} and the induction on backward decoding,  the decoding is successful with high probability if
\begin{IEEEeqnarray}{rCl}
R_1''&\leq& I(X;Y_3,V|X_3,U)-\delta(\epsilon/2)\nonumber\\
&\stackrel{(a)}=& I(X;Y_3|U,V,X_3)-\delta(\epsilon/2)\label{eq:th1.1}\\
R_1+R_2&\leq& I(X,U;Y_3|V,X_3)-\delta(\epsilon/2),\nonumber\\
&\stackrel{(b)}=& I(X;Y_3|V,X_3)-\delta(\epsilon/2)
\end{IEEEeqnarray}
where  (a) and (b) follow  from Markov chains  $V-(U,X_3)-X$ and $U-(V,X,X_3)-Y_3$, respectively.

\subsubsection{Decoding}
Receiver  2 applies backward decoding to decode $v_b^n$ and $u_b^n$,  for $b\in[1:B]$. Specifically, after $(B+1)$-block transmission, assuming receiver 2 already decodes $(m'_{1,b}, m_{2,b})$ based on $y^n_{2,b+1}$, it looks for a pair of messages $(\hat{m}'_{1,b-1}, \hat{m}_{2,b-1})$ such that 
\begin{IEEEeqnarray}{rCl}
&&\big(u_b^n({m}_{2,b},{m}'_{1,b}|\hat{m}_{2,b-1},\hat{m}'_{1,b-1}),\nonumber\\&&v_{b}^n(\hat{m}_{2,b-1},\hat{m}'_{1,b-1}),y_{2,b}^n
\big) \in\mathcal{T}^n_{\epsilon}(p_{UVY_2}).\nonumber
\end{IEEEeqnarray}
By the independence of the codebooks, the Markov lemma, packing lemma and the induction on backward decoding,  the decoding is successful with high probability if
\begin{IEEEeqnarray}{rCl}\label{}
R_1'+R_2\leq I(U,V;Y_2)-\delta(\epsilon)
\end{IEEEeqnarray} 

Receiver  1 applies backward decoding to decode $u_b^n,v_b^n,x_{3,b}^n$ and $x_b^n$, for $b\in[1:B+1]$. Specifically, after $(B+1)$-block transmission, assuming receiver 1 already decodes $({m}''_{1,b},{m}_{2,b},{m}'_{1,b})$ based on $y^n_{1,b+1}$, it looks for a tuple of messages $(\hat{m}''_{1,b-1},\hat{m}_{2,b-1},\hat{m}'_{1,b-1})$ such that 
\begin{IEEEeqnarray}{rCl}\label{}
&&\big( u_b^n({m}_{2,b},{m}'_{1,b}|\hat{m}
_{2,b-1},\hat{m}'_{1,b-1}),\nonumber\\&&v_{b}^n(\hat{m}_{2,b-1},\hat{m}'_{1,b-1}), x_{3,b}^n(\hat{m}''_{1,b-1}|\hat{m}_{2,b-1},\hat{m}'_{1,b-1}),\nonumber\\&&
x_{b}^n({m}''_{1,b}|{m}_{2,b},{m}'_{1,b},\hat{m}''_{1,b-1},\hat{m}_{2,b-1},\hat{m}'_{1,b-1}),y_{1,b}^n
\big)\nonumber\\&& \in\mathcal{T}^n_{\epsilon}(p_{UVXX_3Y_1}).\nonumber
\end{IEEEeqnarray}

By the independence of the codebooks, the Markov lemma, packing lemma and the induction on backward decoding,  the decoding is successful with high probability if
\begin{IEEEeqnarray}{rCl}\label{eq:th1.9}
R_1+R_2&\leq& I(U,V,X,X_3;Y_1)-\delta(\epsilon)\nonumber\\
&\stackrel{(a)}=& I(X,X_3;Y_1)-\delta(\epsilon),\\
R''_1&\leq& I(X,X_3;Y_1|U,V)-\delta(\epsilon)\label{eq:th1.9}
\end{IEEEeqnarray} 
where (a) holds because of Markov chain $(U,V)-(X,X_3)-Y_1$.

Combining (\ref{eq:th1.1}--\ref{eq:th1.9}), letting $\epsilon\to 0$, and using Fourier-Motzkin elimination \cite{Gamal'book} to eliminate ${R}'_1$ and $R_1''$, we obtain the inner bounds as below.
 \begin{subequations}
  \begin{IEEEeqnarray}{rCl}\label{eq:regionThm1temp}
R_2&\leq& I(U,V;Y_2),\label{eq:1}\\
R_1+R_2&\leq& I(X,X_3;Y_1),\label{eq:usless1}\\
R_1+R_2&\leq& I(X;Y_3|V,X_3),\\
R_1+R_2&\leq& I(U,V;Y_2)+I(X;Y_3|U,V,X_3),\\
R_1+R_2&\leq& I(U,V;Y_2)+I(X,X_3;Y_1|U,V).\label{eq:2}
\end{IEEEeqnarray}
\end{subequations}
Notice that $I(U,V;Y_2)\leq I(U,V;Y_1)$ and by  \eqref{eq:2}, we have $R_1+R_2\leq I(U,V;Y_2)+I(X,X_3;Y_1|U,V)\leq I(X,X_3;Y_1)$, which makes constraint \eqref{eq:usless1} invalid.

\subsection{Inner bound for  Type-II PDRBC}\label{sec:InType2}
We present a block-Markov coding scheme that consists of $B+1$ blocks, where messages $M_{1,b}\in[1:2^{nR_1}]$ and $M_{2,b}\in[1:2^{nR_2}]$, for $b\in[1:B]$, are sent to the receivers over $B+1$ blocks.  Let messages $M_{1,B+1}=M_{2,B+1}=1$ and assume   $(M_{1,B+1},M_{2,B+1})$  are known by the relay and receivers before communication.

\subsubsection{Codebook}
Fix the pmf $p(x_3)p(u|x_3)p(x|u)$.
  For each block $b\in[1:B+1]$, randomly and independently generate $2^{nR_{2}}$ sequences $x_{3,b}^n(m_{2,b-1})\sim \prod^n_{i=1}p_{X_3}(x_{3,b,i})$, $m_{2,b-1}\in [1:2^{nR_2}]$.  For each $m_{2,b-1}$, randomly and independently generate $2^{n{R}_2}$ sequences ${u}_{b}^n(m_{2,b}|m_{2,b-1})\sim \prod^n_{i=1}p_{U|X_3}(u_{b,i}|x_{3,b,i})$. For each $(m_{2,b},m_{2,b-1})$, randomly and independently generate $2^{n{R}_1}$ sequences ${x}_{b}^n(m_{1,b}|m_{2,b},m_{2,b-1})\sim \prod^n_{i=1}p_{X|U}(x_{b,i}|u_{b,i})$.


\subsubsection{Transmitter encoding}
  In each block $b\in[1:B+1]$,the transmitter  sends $x_{b}(m_{1,b}|m_{2,b},m_{2,b-1}).$


\subsubsection{Relay encoding}
In each block $b\in[1:B+1]$, the relay decodes $u_b^n$ for $b\in[1:B]$, by looking for a tuple of messages $\hat{m}_{2,b}$ such that
\begin{IEEEeqnarray}{rCl}\label{eq:Sch1CF}
\big( u_b^n(\hat{m}_{2,b},|m_{2,b-1}), x_{3,b}^n(m_{2,b-1}),y_{3,b}^n
\big) \in\mathcal{T}^n_{\epsilon/2}(p_{UX_3Y_3}).\nonumber
\end{IEEEeqnarray}

By the independence of the codebooks, the Markov lemma, packing lemma  and the induction on backward decoding,  the decoding is successful with high probability if
\begin{IEEEeqnarray}{rCl}
R_2&\leq& I(U;Y_3|X_3)-\delta(\epsilon/2).
\end{IEEEeqnarray}

\subsubsection{Decoding}
Receiver  2 applies backward decoding to decode $u_b^n$ and $x_{3,b}^n$,  for $b\in[1:B]$. Specifically, after $(B+1)$-block transmission, assuming Receiver 2 already decodes $m_{2,b}$ based on $y^n_{2,b+1}$, it looks for a message  $\hat{m}_{2,b-1}$ such that 
\begin{IEEEeqnarray}{rCl}
\big(u_b^n({m}_{2,b}|\hat{m}_{2,b-1}),x_{3,b}^n(\hat{m}_{2,b-1}),y_{2,b}^n
\big) \in\mathcal{T}^n_{\epsilon}(p_{UX_3Y_2}).\nonumber
\end{IEEEeqnarray}
By the independence of the codebooks, the Markov lemma, packing lemma and the induction on backward decoding,  the decoding is successful with high probability if
\begin{IEEEeqnarray}{rCl}\label{}
R_2&\leq& I(U,X_3;Y_2)-\delta(\epsilon)
\end{IEEEeqnarray}

Receiver  1 applies backward decoding to decode $u_b^n,x_{3,b}^n$ and $x_b^n$, for $b\in[1:B+1]$. Specifically, after $(B+1)$-block transmission, assuming Receiver 1 already decodes ${m}_{2,b}$ based on $y^n_{1,b+1}$, it looks for a pair of messages $(\hat{m}_{1,b},\hat{m}_{2,b-1})$ such that 
\begin{IEEEeqnarray}{rCl}
&&\big(u_b^n({m}_{2,b}|\hat{m}_{2,b-1}),x_{3,b}^n(\hat{m}_{2,b-1}),\nonumber\\&&{x}_{b}^n(\hat{m}_{1,b}|m_{2,b},\hat{m}_{2,b-1}),y_{1,b}^n
\big) \in\mathcal{T}^n_{\epsilon}(p_{UX_3XY_1}).\nonumber
\end{IEEEeqnarray}

By the independence of the codebooks, the Markov lemma, packing lemma and the induction on backward decoding,  the decoding is successful with high probability if
\begin{IEEEeqnarray}{rCl}\label{}
R_1&\leq& I(X;Y_1|U,X_3)-\delta(\epsilon)\\
R_1+R_2&\leq& I(X,X_3;Y_1)-\delta(\epsilon).
\end{IEEEeqnarray}

Combining (\ref{eq:th1.1}--\ref{eq:th1.9}) and letting $\epsilon\to 0$, we obtain the inner bounds as below
\begin{subequations}
\begin{IEEEeqnarray}{rCl}
R_2&\leq &I(U,X_3;Y_2),\label{eq:typeII1}\\
R_2&\leq& I(U;Y_3|X_3),\\
R_1&\leq& I(X;Y_1|U,X_3),\label{eq:typeII3}\\
R_1+R_2&\leq& I(X,X_3;Y_1)\label{eq:typeII3.1}
\end{IEEEeqnarray}
\end{subequations}
for some pmf $p(x_3)p(u|x_3)p(x|u)$.
Notice that from \eqref{eq:typeII1} and \eqref{eq:typeII3}, we have
\begin{IEEEeqnarray}{rCl} \label{eq:typeII4}
R_1+R_2&\leq& I(U,X_3;Y_2)+I(X;Y_1|U,X_3)\\
&\stackrel{(a)}\leq& I(U,X_3;Y_1)+I(X;Y_1|U,X_3)\\
&= & I(U,X,X_3;Y_1)\\
&\stackrel{(b)}=& I(X,X_3;Y_1)
\end{IEEEeqnarray}
where (a) holds by Markov chain $(U,X_3)-Y_1-Y_2$ and processing inequality; (b) follows from Markov chain $U-(X,X_3)-Y_1$. Thus, rate constraint \eqref{eq:typeII3.1} is invalid given \eqref{eq:typeII1} and \eqref{eq:typeII3}, which leads to the  inner bounds as shown in Theorem \ref{Them:type2}.

\section{Outer Bounds for PDRBC}\label{sec:OutTDM} 

\subsection{Outer bound for Type-I PDRBC}\label{sec:OutType1}
Define 
\begin{IEEEeqnarray}{rCl}\label{eq:defin1}
U_i=(M_2,Y^{i-1}_1,Y_2^{i-1}),V_i=(Y_1^{i-1},Y^{i-1}_2)\quad 
\end{IEEEeqnarray}
and let $\epsilon_n$ tend to 0 as $n\to\infty$. Introduce   a time-sharing random variable $Q$ that is uniformly distributed over $[1:n]$ and  independent of $(M_1, M_2, U^n,V^n,X^n, X_3^n,Y_1^n, Y^n_2,Y_3^n)$.

By Fano's inequality  we have 
\begin{IEEEeqnarray}{rCl}
nR_1
&\leq& I(M_1;Y^n_1,Y^n_2|M_2)+n\epsilon_n\nonumber\\
 &\stackrel{(a)}=&\sum^n_{i=1} I(X_i,M_1;Y_{1,i},Y_{2,i}|U_i,V_i)+n\epsilon_n\nonumber\\
 &\leq&\sum^n_{i=1} I(X_i,X_{3,i},M_1;Y_{1,i},Y_{2,i}|U_i,V_i)+n\epsilon_n\nonumber\\
 & \stackrel{(b)}= &\sum^n_{i=1} I(X_i,X_{3,i};Y_{1,i},Y_{2,i}|U_i,V_i)+n\epsilon_n\nonumber\\
 &=&nI(X_Q,X_{3,Q};Y_{1,Q},Y_{2,Q}|U_Q,V_Q,Q)+n\epsilon_n\quad
\end{IEEEeqnarray}
where (a) holds by the definition of $U_i,V_i$ and since $X_i$ is a function of $M_1$ and $M_2$; (b) follows from Markov chain $(U_i,V_i,M_1)-(X_i,X_{3,i})-(Y_{1,i},Y_{2,i})$. 

Similarly, 
 \begin{IEEEeqnarray}{rCl}\label{outertypeIstart}
&&n(R_1+R_2)\nonumber\\
&&\leq I(M_1,M_2;Y^n_1,Y^n_2,Y_3^n)+n\epsilon_n\nonumber\\
&&\stackrel{(a)}=\sum^n_{i=1}I(M_1,M_2;Y_{1,i},Y_{2,i},Y_{3,i}|V_i,Y^{i-1}_3)+n\epsilon_n\nonumber\\
&& \stackrel{(b)}= H(Y_{1,i},Y_{2,i},Y_{3,i}|V_i,X_{3,i},Y_3^{i-1})\nonumber\\
&&\quad-H(Y_{1,i},Y_{2,i},Y_{3,i}|X_i,X_{3,i},V_i)+n\epsilon_n\nonumber\\
&&\leq \sum^n_{i=1}I(X_i;Y_{1,i},Y_{2,i},Y_{3,i}|X_{3,i},V_i)+n\epsilon_n\nonumber\\
&&\stackrel{(c)}=  \sum^n_{i=1}I(X_i;Y_{3,i}|X_{3,i},V_i)+n\epsilon_n\nonumber\\
&& =nI(X_Q;Y_{3,Q}|X_{3,Q},V_Q,Q)+n\epsilon_n
\end{IEEEeqnarray}
where   (a) follows by the definition of $V_i$;  (b) holds  because $X_{3,i}$ is a function of $Y_3^{i-1}$, $X_i$ is a function of $M_1$ and $M_2$ and $(M_1,M_2,V_i,Y^{i-1}_3)-(X_i,X_{3,i})-(Y_{1,i},Y_{2,i},Y_{3,i})$ forms Markov chain; (c) follows from the property of Type-I PDRBC, which has  Markov chain $X_i-(X_{3,i},Y_{3,i},V_i)-(Y_{1,i},Y_{2,i})$.

And, 
\begin{IEEEeqnarray}{rCl}
nR_2&\leq& I(M_2;Y^n_2)+n\epsilon_n\nonumber\\
&=&\sum^n_{i=1}I(M_2;Y_{2,i}|Y^{i-1}_2)+n\epsilon_n\nonumber\\
&\leq& \sum^n_{i=1}I(M_2,Y^{i-1}_2;Y_{2,i})+n\epsilon_n\nonumber\\
&\leq & \sum^n_{i=1}I(M_2,Y^{i-1}_2,Y_1^{i-1};Y_{2,i})+n\epsilon_n\nonumber\\
& =& \sum^n_{i=1}I(U_i,V_i;Y_{2,i})+n\epsilon_n\nonumber\\
&=& nI(U_Q,V_Q;Y_{2,Q}|Q)+n\epsilon_n\nonumber\\
&\leq&nI(U_Q,V_Q,Q;Y_{2,Q})+n\epsilon_n
\end{IEEEeqnarray}

Also,
  \begin{IEEEeqnarray}{rCl}\label{outertypeIend}
&&nR_1\nonumber\\
&&\leq I(M_1;Y^n_1,Y^n_2,Y_3^n|M_2)+n\epsilon_n\nonumber\\
&&\stackrel{(a)}=\sum^n_{i=1}I(X_i,M_1;Y_{1,i},Y_{2,i},Y_{3,i}|U_i,V_i,X_{3,i},Y_3^{i-1})+n\epsilon_n\nonumber\\
&&\leq \sum^n_{i=1}I(X_i;Y_{1,i},Y_{2,i},Y_{3,i}|U_i,V_i,X_{3,i})+n\epsilon_n\nonumber\\
&& \stackrel{(b)}= \sum^n_{i=1}I(X_i;Y_{3,i}|U_i,V_i,X_{3,i})+n\epsilon_n\epsilon_n\nonumber\\
&& =nI(X_Q;Y_{3,Q}|U_Q,V_Q,X_{3,Q},Q)+n\epsilon_n
\end{IEEEeqnarray}
where (a) holds by the definition of $(U_i,V_i)$ and since $X_{3,i}$ is a function of $Y_{3}^{i-1}$;   (b) holds due to holds due to the property of Type-I PDRBC, which has Markov chain $X_i-(U_i,V_i,X_{3,i},Y_{3,i})-(Y_{1,i},Y_{2,i})$.

 Define 
$U=(Q,U_Q),V=(Q,V_Q), X=X_Q,X_3=X_{3Q},Y_1=X_{1Q},Y_2=Y_{2Q}$ and $Y_3=Y_{3Q}$.  
By the definition of $(U_i,V_i)$ in \eqref{eq:defin1}, we have  $V=(Q,V_Q)=(Q,Y^{Q-1}_1,Y_2^{Q-1})$ and $U=(Q,U_Q)=(Q,M_2,Y^{Q-1}_1,Y_2^{Q-1})$, and thus Markov chains
$X_3-U-V$ and $V-(X_3,U)-X$ hold, leading to  pmf $p(u,v)p(x_3|u,v)p(x|x_3,u)$ for this outer bound.

Since  $\epsilon_n$ tends to 0 as $n\to \infty$, combing (\ref{outertypeIstart}--\ref{outertypeIend}) we obtain an outer bound as shown in Theorem \ref{eq:regionThm1}. 

\subsection{Outer bound  for  Type-II PDRBC}\label{sec:OutType2}
Define 
\begin{IEEEeqnarray}{rCl}\label{eq:defin2}
U_i=(M_2,Y^{i-1}_1,Y_2^{i-1},Y_3^{i-1}). 
\end{IEEEeqnarray}
Introduce   a time-sharing random variable $Q$ that is uniformly distributed over $[1:n]$ and  independent of $(M_1, M_2, U^n,X^n, X_3^n,Y_1^n, Y^n_2,Y_3^n)$.

By Fano's inequality  we have 
 \begin{IEEEeqnarray}{rCl}\label{rCl}\label{outertype2start}
nR_1
&\leq& I(M_1;Y^n_1,Y^n_2,Y_3^n|M_2)+n\epsilon_n\nonumber\\
&\stackrel{(a)}=&\sum^n_{i=1}I(X_i,M_1;Y_{1,i},Y_{2,i},Y_{3,i}|U_i,X_{3,i},)+n\epsilon_n\nonumber\\
& \stackrel{(b)}=& \sum^n_{i=1}I(X_i;Y_{1,i},Y_{2,i},Y_{3,i}|U_i,X_{3,i})+n\epsilon_n\nonumber\\
& \stackrel{(c)}=& \sum^n_{i=1}I(X_i;Y_{1,i}|U_i,X_{3,i})+n\epsilon_n\nonumber\\
&=&nI(X_Q;Y_{1,Q}|U_Q,X_{3,Q},Q)+n\epsilon_n
\end{IEEEeqnarray}
where (a) holds because $X_{3,i}$ is a function of $Y_{3}^{i-1}$ and $X_i$ is a function of $(M_1,M_2)$;  (b) follows from the Markov chain $M_1-(U_i,X_i,X_{3,i})-(Y_{1,i},Y_{2,i},Y_{3,i})$; (c) holds due to the property of Type-II PDRBC, which has Markov chain $X_i-(U_i,X_{3,i},Y_{1,i})-(Y_{2,i},Y_{3,i})$.

And, 
\begin{IEEEeqnarray}{rCl}\label{rCl}
nR_2&\leq& I(M_2;Y^n_2)+n\epsilon_n\nonumber\\
&=&\sum^n_{i=1}I(M_2;Y_{2,i}|Y^{i-1}_2)+n\epsilon_n\nonumber\\
&\leq& \sum^n_{i=1}I(M_2,Y_1^{i-1},Y_2^{i-1},Y_3^{i-1};Y_{2,i})+n\epsilon_n\nonumber\\
& =&  \sum^n_{i=1}I(U_i,X_{3,i};Y_{2,i})+n\epsilon_n\nonumber\\
&=&    nI(U_Q,X_{3,Q};Y_{2,Q}|Q)+n\epsilon_n\nonumber\\
&\leq& nI(U_Q,X_{3,Q},Q;Y_{2,Q})+n\epsilon_n
\end{IEEEeqnarray}

Also,
\begin{IEEEeqnarray}{rCl}\label{outertype2end}
nR_2
&\leq& I(M_2;Y^n_2,Y^n_3)+n\epsilon_n\nonumber\\
 &\stackrel{(a)}=& \sum^n_{i=1}
 I(M_2;Y_{2,i},Y_{3,i}|Y_2^{i-1},Y_3^{i-1},X_{3,i})+n\epsilon_n\nonumber\\
 &\leq& \sum^n_{i=1}I(M_2,Y_1^{i-1},Y_2^{i-1},Y_3^{i-1};Y_{2,i},Y_{3,i}|X_{3,i})+n\epsilon_n\nonumber\\
 &\stackrel{(b)}= &\sum^n_{i=1} I(U_i;Y_{2,i},Y_{3,i}|X_{3,i})+n\epsilon_n\nonumber\\
 & \stackrel{(c)}=& \sum^n_{i=1} I(U_i;Y_{3,i}|X_{3,i})+n\epsilon_n\nonumber\\
 & =& nI(U_Q;Y_{3,Q}|X_{3,Q},Q)+n\epsilon_n\nonumber\\
 &\leq&   nI(U_Q,Q;Y_{3,Q}|X_{3,Q})+n\epsilon_n
\end{IEEEeqnarray}
where (a) holds  since $X_{3,i}$ is a function of $Y_3^{i-1}$; (b) holds by the definition of $U_i$; (c) holds due to the property of Type-II PDRBC, which has Markov  chain $U_i-(X_{3,i},Y_{3,i})-Y_{2,i}$.

 Define 
$U=(Q,U_Q), X=X_Q,X_3=X_{3Q},Y_1=X_{1Q},Y_2=Y_{2Q}$ and $Y_3=(Q,Y_{3Q})$. By the definition of $U_i$ in \eqref{eq:defin2}, we have  $U=(Q,U_Q)=(Q,M_2,Y^{Q-1}_1,Y_2^{Q-1},Y_3^{Q-1})$, and since $X_3=X_{3,Q}$ is a function of $Y_3^{Q-1}$, we have Markov chain $X_3-U-X$ leading to  pmf $p(u,x_3)p(x|u)$ for this outer bound.

Since  $\epsilon_n$ tends to 0 as $n\to \infty$, combing (\ref{outertype2start}--\ref{outertype2end}) we obtain an outer bound same as the inner bound of Theorem \ref{Them:type2}.

\subsection{Outer bound for Type-III PDRBC}\label{sec:OutType3}
Define 
\begin{IEEEeqnarray}{rCl}\label{eq:defin3}
U_i=(M_2,Y_2^{i-1},Y_3^{i-1}). 
\end{IEEEeqnarray}
Introduce   a time-sharing random variable $Q$ that is uniformly distributed over $[1:n]$ and  independent of $(M_1, M_2, X^n, X_3^n,Y_1^n, Y^n_2,Y_3^n)$.

 By Fano's inequality  we have
  \begin{IEEEeqnarray}{rCl}\label{outertype3end}
&&nR_1\nonumber\\
&&\leq I(M_1;Y^n_1)+n\epsilon_n\nonumber\\
&&\stackrel{(a)}=\sum^n_{i=1}I(X_i,M_1;Y_{1,i},Y_{2,i},Y_{3,i}|U_i,X_{3,i})+n\epsilon_n\nonumber\\
&& \stackrel{(b)}= \sum^n_{i=1}I(X_i;Y_{1,i},Y_{2,i},Y_{3,i}|U_i,X_{3,i})+n\epsilon_n\nonumber\\
&& \stackrel{(c)}= \sum^n_{i=1}I(X_i;Y_{1,i}|U_i,X_{3,i})+n\epsilon_n\nonumber\\
&&= nI(X_Q;Y_{1,Q}|U_Q,X_{3,Q},Q)+n\epsilon_n
\end{IEEEeqnarray}
where (a) holds since $X_{3,i}$ is a function of $Y_{3}^{i-1}$ and $X_i$ is a function of $(M_1,M_2)$;  (b) follows from the Markov chain $M_1-(U_i,X_i,X_{3,i})-(Y_{1,i},Y_{2,i},Y_{3,i})$; (c) holds due to the property of Type-III PDRBC, which has Markov chain $X_i-(U_i,X_{3,i},Y_{1,i})-(Y_{2,i},Y_{3,i})$.

Also,
 \begin{IEEEeqnarray}{rCl}\label{outertype3start}
nR_2 
&\leq& I(M_2;Y^n_2,Y^n_3)+n\epsilon_n\nonumber\\
&=& 
 I(M_2;Y_{2,i},Y_{3,i}|Y_2^{i-1},Y_3^{i-1})+n\epsilon_n\nonumber\\
 &=& I(M_2;Y_{2,i},Y_{3,i}|Y_2^{i-1},Y_3^{i-1},X_{3,i})+n\epsilon_n\nonumber\\
 &\stackrel{(a)}\leq& I(U_i;Y_{2,i},Y_{3,i}|X_{3,i})+n\epsilon_n\nonumber\\
 & \stackrel{(b)}=&  I(U_i;Y_{2,i}|X_{3,i})+n\epsilon_n\nonumber\\
 &\leq &  I(U_Q,Q;Y_{2,Q}|X_{3,Q})+n\epsilon_n
\end{IEEEeqnarray}
where (a) holds by the definition of $U_i$; (b) holds due to the property of Type-III PDRBC, which has Markov  chain $U_i-(X_{3,i},Y_{2,i})-Y_{3,i}$.

 Define 
$U=(Q,U_Q), X=X_Q,X_3=X_{3Q},Y_1=X_{1Q},Y_2=Y_{2Q}$ and $Y_3=(Q,Y_{3Q})$. By the definition of $U_i$ in \eqref{eq:defin3}, we have  $U=(Q,U_Q)=(Q,M_2,Y_2^{Q-1},Y_3^{Q-1})$, and since $X_3=X_{3,Q}$ is a function of $Y_3^{Q-1}$, we have Markov chain $X_3-U-X$, leading to  pmf $p(u,x_3)p(x|u)$ for this outer bound.

Since  $\epsilon_n$ tends to 0 as $n\to \infty$, combing (\ref{outertype3start}--\ref{outertype3end}), we obtain the outer bound as below.
\begin{subequations}
\begin{IEEEeqnarray*}{rCl}
R_2&\leq& I(U;Y_{2}|X_3)=\sum_{x_3\in\set{X}_3} I(U;Y_{2}|X_3=x_3)\\
R_1&\leq& I(X;Y_1|U,X_3)=\sum_{x_3\in\set{X}_3} I(X;Y_1|U,X_3=x_3)
\end{IEEEeqnarray*}
\end{subequations}
for some pmf $p(u,x_3)p(x|u)$.
Note that $I(U;Y_{2}|X_3=x_3)$ and $I(X;Y_1|U,X_3=x_3)$ both are linear functions of $p(x,u)$, and since $X_3$ is simplex, the boundary points on the outer bound
\[\sum_{x_3\in\set{X}_3} \big(I(U;Y_{2}|X_3=x_3)+I(X;Y_1|U,X_3=x_3)\big)\]
 is maximized at an extreme point. Thus, the corresponding outer bound  can be  characterized   as
 \begin{subequations}
\begin{IEEEeqnarray*}{rCl}
R_2&\leq& I(U;Y_{2}|X_3)= I(U;Y_{2}|X_3=x_3)\\
R_1&\leq& I(X;Y_1|U,X_3)= I(X;Y_1|U,X_3=x_3)
\end{IEEEeqnarray*}
\end{subequations}
for some value $x_3\in\set{X}_3$ and pmf $p(x,u)$, which completes the converse.

  \section{Proof of Theorem \ref{Them:Gaussian}}\label{Sec:proofGaussian}
  A rigorous proof that our  results in Theorem   \ref{Them:type1}, \ref{Them:type2} and \ref{Them:type3} hold also for the Gaussian PDRBC  is omitted for brevity. 
  In the following subsections, we will first prove the achievability of rate region \eqref{region:T1}, \eqref{region:T2} and \eqref{region:T3}, and then show that these inner bounds are tight.
\subsection{Capacity region on  Type-I Gaussian  PDRBC}\label{Proof:GaussianT1}
 From the inner bound    of Theorem \ref{Them:type1}, we obtain a potentially smaller inner bound:
\begin{subequations}\label{GaussianEvaluate}
  \begin{IEEEeqnarray}{rCl}
R_1&\leq& I(X;Y_3|U,X_3), \\
R_1&\leq& I(X,X_3;Y_1|U,V),\\
R_2&\leq& I(U,V;Y_2),\\
R_1+R_2&\leq& I(X;Y_3|X_3,V)
\end{IEEEeqnarray}
\end{subequations}
for some probability density function (pdf) $p(u,v)p(x_3|v)p(x|u,x_3)$. 

Now let $V\sim\set{N}(0,\bar{\theta}_rP_r)$, $W_0\sim\set{N}(0,{\theta}_rP_r)$, $W_1\sim\set{N}(0,\bar{\beta}\bar{\theta}P)$, $W_2\sim\set{N}(0,\bar{\alpha}{\theta}P)$ and 
\begin{IEEEeqnarray*}{rCl}
U&=&\rho_1 V+W_1,~X_3=V+W_0,~X=U+\rho_2 X_3+W_2
\end{IEEEeqnarray*}
where $(U,V,W_0,W_1,W_2)$ are  auxiliary random variables,  independent of each other, and 
\begin{IEEEeqnarray*}{rCl}
\rho_2= \sqrt{\frac{\alpha\theta P}{\theta_r P_r}}, \quad \rho_1+\rho_2=\sqrt{\frac{\beta\bar{\theta} P}{\bar{\theta_r} P_r}}
\end{IEEEeqnarray*}
with $0\leq\alpha,\beta,\theta,\theta_r\leq1$.

With the choice above, we have
\begin{IEEEeqnarray*}{rCl}
&&I(X;Y_3|U,X_3)= C\left(\bar{\alpha}\theta\frac{P}{\sigma_3^2}\right),\\
&&I(X,X_3;Y_1|U,V)= C\left( \frac{\theta P+\theta_r P_r+2\sqrt{\alpha\theta \theta_r  PP_r}}{\sigma_1^2} \right),\\
&&I(U,V;Y_2)=C\left(\frac{\bar{\theta}P+\bar{\theta}_rP_r+2\sqrt{\beta\bar{\theta}\bar{\theta}_rPP_r}}{\theta P+\theta_r P_r+2\sqrt{\alpha\theta \theta_r  PP_r}+\sigma_2^2}\right),\\
&& I(X;Y_3|X_3,V) = C\left(\frac{\bar{\beta}\bar{\theta}P+\bar{\alpha}\theta  P}{\sigma_3^2}\right).
\end{IEEEeqnarray*}
Thus, the inner bound of Theorem  \ref{Them:type1} for the Gaussian case consists of all rate pairs $(R_1,R_2)$ satisfying
\begin{subequations}\label{eq:type1in}
\begin{IEEEeqnarray}{rCl}
R_2&\leq& C\left(\frac{\bar{\theta}P+\bar{\theta}_rP_r+2\sqrt{\beta\bar{\theta}\bar{\theta}_rPP_r}}{\theta P+\theta_r P_r+2\sqrt{\alpha\theta \theta_r  PP_r}+\sigma_2^2}\right),\quad\\
R_1&\leq& \min \left \{C\left(\bar{\alpha}\theta\frac{P}{\sigma_3^2}\right), \right. \nonumber\\&&\left. \quad C\left( \frac{\theta P+\theta_r P_r+2\sqrt{\alpha\theta \theta_r  PP_r}}{\sigma_1^2} \right)\right \},\\
R_1&+&R_2 \leq C\left(\frac{\bar{\beta}\bar{\theta}P+\bar{\alpha}\theta  P}{\sigma_3^2}\right)
\end{IEEEeqnarray}
\end{subequations}
with $0\leq\alpha,\beta,\theta,\theta_r\leq1$. 

Denote the  inner bound  in \eqref{eq:type1in} as $\set{C}_{\text{inner,Wu}}$, and compare it with Bhaskaran's inner bound $\set{C}_{\text{inner,Bha}}$ in \cite[Theorem 2]{Bha'08}, we find that both inner bounds have same rate expression, except that $\set{C}_{\text{inner,Bha}}$ has an additional rate constraint $R_2\leq C\left(\frac{\bar{\beta}\bar{\theta}P}{\sigma^2}\right)$. Thus, we have  $$\set{C}_{\text{inner,Bha}}\subseteq\set{C}_{\text{inner,Wu}}.$$ 

In  \cite{Bha'08} it shows that $\set{C}_{\text{{inner,Bha}}}$ is equivalent to 
\begin{subequations}
\begin{IEEEeqnarray*}{rCl}
R_2&\leq& C\left( \frac{\bar{\delta}[P+P_r+2\sqrt{\gamma PP_r}]}{{\delta}[P+P_r+2\sqrt{\gamma PP_r}]+\sigma_2^2} \right), \\
R_1&\leq& C\left(\frac{{\delta}[P+P_r+2\sqrt{\gamma PP_r}]}{\sigma_1^2} \right),\\
R_1+R_2 &\leq& C\left(\frac{\bar{\gamma} P}{\sigma_3^2}\right), \quad 0\leq\gamma,\delta\leq 1,
\end{IEEEeqnarray*}
\end{subequations}
which  matches the outer bound of the capacity region for the Type-I Gaussian PDRBC, i.e., $\set{C}_{\text{inner,Bha}}=\set{C}_{\text{PD}}$. Thus, we conclude that $\set{C}_{\text{inner,Wu}}$ is tight for the Type-I Gaussian PDRBC.

\subsection{Capacity region on Type-II Gaussian PDRBC}\label{sec:type2Gaussian}

\begin{itemize}
\item[1)] \textit{Proof of the Achievability} \\
Let 
\begin{IEEEeqnarray*}{rCl}
U=\rho_1 X_3+W_1, \quad X=\rho_2 U+W_2
\end{IEEEeqnarray*}
where $(X_3,W_1,W_2)$ are independent with each other and $X_3\sim\set{N}(0,P_r)$, $W_1\sim\set{N}(0,\beta P)$, $(\rho_2)^2U\sim\set{N}(0,\bar{\alpha}P)$, $W_2\sim\set{N}(0,\alpha P)$, with $0\leq \rho_1,\rho_2,\alpha,\beta\leq 1$. 
With the choice above, we obtain
\begin{subequations}
\begin{IEEEeqnarray*}{rCl}
R_2&\leq& I(U,X_3;Y_2)=C\left(\frac{\bar{\alpha}P+P_r}{\alpha P+\sigma^2_2} \right) ,\\
R_2&\leq&I(U;Y_3|X_3)= C\left(\frac{(\beta-{\alpha})P}{\alpha P+\sigma^2_3} \right),\\
R_1&\leq&I(X;Y_1|U,X_3)= C\left(\frac{\alpha P}{\sigma_1^2} \right)
\end{IEEEeqnarray*}
\end{subequations}
where $0\leq \rho_1,\rho_2,\alpha,\beta\leq 1$ and $\beta\geq{\alpha}$. 
\item[2)] \textit{Proof of the Converse}\\
Consider 
\begin{IEEEeqnarray*}{rCl}
I(U;Y_3|X_3) = h(Y_3|X_3)-h(Y_3|X_3,U)
\end{IEEEeqnarray*}
Since
\begin{equation*}
\begin{aligned}
\frac{1}{2}\log \left(2\pi e\sigma_3^2\right)&=h(Z_3)\leq h(Y_3|X_3)\leq h(X+Z_3)\\&\leq \frac{1}{2}\log \left( 2\pi e(P+\sigma_3^2)\right),
\end{aligned}
\end{equation*} 
there must exist a $\beta\in[0,1]$ such that
\[h(Y_3|X_3)=\frac{1}{2}\log \left( 2\pi e(\beta P+\sigma_3^2)\right)\]
Similarly, since 
\begin{equation*}
\begin{aligned}
  \frac{1}{2}\log \left(2\pi e\sigma_3^2\right)&= h(Z_3) \leq h(Y_3|U,X_3)\leq h(Y_3|X_3) \\&=\frac{1}{2}\log \left( 2\pi e(\beta P+\sigma_3^2)\right)
\end{aligned}
\end{equation*}
there must exist an ${\alpha}\in[0,\beta]$ such that
\begin{IEEEeqnarray}{rCl}\label{eq:HY3UX3}
 h(Y_3|U,X_3)=\frac{1}{2}\log \left( 2\pi e({\alpha} P+\sigma_3^2)\right).
\end{IEEEeqnarray}
Thus,
\[R_2\leq I(U;Y_3|X_3)= C\left( \frac{(\beta-{\alpha})P}{{\alpha} P+\sigma_3^2}\right). \]
Next consider
\begin{IEEEeqnarray*}{rCl}
&&I(U,X_3;Y_2)=h(Y_2)-h(Y_2|U,X_3)\nonumber\\
&&\quad\leq \frac{1}{2}\log \left( 2\pi e(P+P_r+\sigma_2^2)\right)- h(Y_2|U,X_3).
\end{IEEEeqnarray*}

By \eqref{eq:HY3UX3} and the conditional EPI in \cite{Gamal'book}, we have
\begin{IEEEeqnarray}{rCl}
h(Y_2|U,X_3)&=&h(Y_3+\tilde{Z}_\textnormal{b}|U,X_3)\nonumber\\
&\geq& \frac{1}{2}\log \left( 2^{2h(Y_3|U,X_3)}+2^{2h(\tilde{Z}_\textnormal{b}|U,X_3)}\right)\nonumber\\
&=&  \frac{1}{2}\log \left( 2\pi e ({\alpha}P+\sigma^2_3) + 2\pi e({\sigma^2_2}-{\sigma^2_3}) \right)\nonumber\\
&=& \frac{1}{2}\log \left( 2\pi e ({\alpha}P+\sigma^2_2) \right).\nonumber
\end{IEEEeqnarray}
Thus,
\[R_2\leq I(U,X_3;Y_2)\leq C\left( \frac{\bar{\alpha} P+P_r}{{\alpha}P+\sigma^2_2}\right). \]
Now consider 
\begin{equation*}
\begin{aligned}
I(X;Y_1|U,X_3)&=h(Y_1|U,X_3)-h(Z_1)\\&=h(Y_1|U,X_3)-\frac{1}{2}\log \left(2\pi\sigma_1^2\right)
\end{aligned}
\end{equation*}
Since
\begin{IEEEeqnarray}{rCl}\label{eq:Y_3UX_3}
&&h(Y_3|U,X_3)\nonumber\\&&\quad=h(Y_1+\hat{Z}_\textnormal{b}|U,X_3)\nonumber\\
&&\quad\geq \frac{1}{2}\log \left( 2^{2h(Y_1|U,X_3)}+2^{2h(\hat{Z}_\textnormal{b}|U,X_3)}\right)\nonumber\\
&&\quad=  \frac{1}{2}\log \left( 2^{2h(Y_1|U,X_3)}+ 2\pi e({\sigma^2_3}-{\sigma_1^2}) \right).
\end{IEEEeqnarray}

Combining \eqref{eq:HY3UX3} with \eqref{eq:Y_3UX_3}, we obtain
\[2\pi e({\alpha} P+\sigma_3^2)\geq 2^{2h(Y_1|U,X_3)}+ 2\pi e({\sigma^2_3}-{\sigma_1^2}).\]
Thus, 
\[h(Y_1|U,X_3)\leq \frac{1}{2} \log \left(2\pi e (\alpha P+\sigma_1^2)\right),\]
which implies
\[R_1\leq I(X;Y_1|U,X_3)\leq C\left( \frac{\alpha P}{\sigma_1^2}\right).\]
This completes the proof of the converse. 
\end{itemize}

\subsection{Capacity region on  Type-III Gaussian  PDRBC}\label{sec:type3Gaussian}
\begin{itemize}
\item[1)] \textit{Proof of the Achievability} \\
The achievability follows by the traditional superposition coding and by shutting down the relay, i.e., set 
 \[X=U+V\]
 where $U\sim\set{N}(0,\bar{\alpha}P)$ and $V\sim\set{N}(0,{\alpha}P)$ are independent of each other. With this choice, it is easy to obtain the rate region in \eqref{region:T2}. 
\item[2)] \textit{Proof of the Converse} \\ 
Consider  
\begin{IEEEeqnarray*}{rCl}
&&I(U;Y_2|X_3)\nonumber\\&&\quad=h(Y_2|X_3)-h(Y_2|U,X_3)\nonumber\\&&\quad\leq \frac{1}{2}\log \left( 2\pi e(P+P_r+\sigma_2^2)\right)-h(Y_2|U,X_3).
\end{IEEEeqnarray*}
Since 
\begin{IEEEeqnarray*}{rCl}
\frac{1}{2}\log \left( 2\pi e\sigma_2^2\right) &=& h(Z_2)\leq h(Y_2|U,X_3)\leq  h(Y_2|X_3)\\ &\leq& \frac{1}{2}\log \left( 2\pi e(P+\sigma_2^2)\right), 
\end{IEEEeqnarray*}
there must exist an $\alpha\in[0,1]$ such that
\begin{IEEEeqnarray}{rCl}\label{eq:type3Y2UX3}
h(Y_2|U,X_3)=\frac{1}{2}\log \left( 2\pi e(\alpha P+\sigma_2^2)\right).
\end{IEEEeqnarray}
Thus 
\[R_2\leq I(U;Y_2|X_3)\leq C\left( \frac{\bar{\alpha}P}{\alpha P+\sigma_2^2}\right).\]
Next consider
\begin{equation*}
\begin{aligned}
I(X;Y_1|U,X_3)&=h(Y_1|U,X_3)-h(Y_1|X,X_3) \\&=h(Y_1|U,X_3)- \frac{1}{2}\log (2\pi e \sigma_1^2).
\end{aligned}
\end{equation*}
Using the conditional EPI, we obtain
\begin{IEEEeqnarray}{rCl}\label{eq:type3Y1UX3}
&&h(Y_2|U,X_3) \nonumber\\&&\quad= h(Y_1+\hat{Z}_\textnormal{c}|U,X_3)\nonumber\\
&&\quad\geq \frac{1}{2}\log \left(2^{2h(Y_1|U,X_3)}+2^{2h(\hat{Z}_\textnormal{c}|U,X_3)} \right)\nonumber\\
&&\quad= \frac{1}{2}\log  \left(2^{2h(Y_1|U,X_3)}+2\pi e(\sigma_2^2-\sigma_1^2) \right)
\end{IEEEeqnarray} 
Combining \eqref{eq:type3Y2UX3} and \eqref{eq:type3Y1UX3}, we have
\begin{equation*}
2\pi e(\alpha P+\sigma_2^2)\geq 2^{2h(Y_1|U,X_3)}+ 2\pi e(\sigma_2^2-\sigma_1^2), 
\end{equation*}
which implies 
\begin{equation*}
h(Y_1|U,X_3)\leq \frac{1}{2}\log \left( 2\pi e (\alpha P+\sigma_1^2) \right),
\end{equation*}
and hence
\[R_1\leq I(X;Y_1|U,X_3)\leq C\left( \frac{\alpha P}{\sigma_1^2}\right).\]
This completes the proof of the converse.

\end{itemize}

\section{conclusion}
Based on different degradation orders among the relay and the receivers' observed signals, we introduce three types of  physically    degraded RBCs.  Theorem 1 presents an inner bound and outer bound on the capacity region when one receiver's output is a degraded form of the other receiver's output, and the stronger receiver's output is a degraded form of the relay's output. These bounds coincide in the rate constraints but under different probability mass functions.  Theorem 2 and 3 establish capacity regions for the PDRBC when one receiver's output is a degraded form of the other receiver's output, and the relay's output is a degraded form of either the stronger or weaker receiver's output. Theorem 4 establishes capacity regions for the Gaussian PDRBCs.

\end{document}